\documentclass[aps,prl,reprint,longbibliography]{revtex4-2}
\usepackage{amsmath,amssymb,amsthm}
\usepackage{orcidlink}
\usepackage{graphicx}
\usepackage{thmtools}
\usepackage{thm-restate}
\usepackage{booktabs}
\usepackage{mathtools}
\usepackage{mathrsfs}
\usepackage{braket}
\usepackage{qcircuit}
\usepackage{placeins}
\usepackage{enumitem}
\usepackage[percent]{overpic}
\newcommand{\abs}[1]{\left|#1\right|}
\newcommand{\norm}[1]{\|#1\|}
\newcommand{\calH}{\mathcal{H}}
\newcommand{\tr}{\mathrm{tr}}
\let\originalsum\sum
\renewcommand{\sum}{\originalsum\nolimits}
\newcommand{\supp}{\mathrm{supp}}

\usepackage{hyperref}

\hypersetup{
    hidelinks,
    colorlinks=true, 
    linkcolor=blue, 
    citecolor=blue, 
    urlcolor=blue, 
   	final=true, 
}
\usepackage[capitalise,nameinlink]{cleveref}
\crefname{fact}{Fact}{Facts}
\newtheorem{proposition}{Proposition}
\newtheorem{lemma}{Lemma}
\newtheorem{theorem}{Theorem}

\newcommand{\rmE}{\mathrm{E}}
\newcommand{\rmH}{\mathrm{H}}
\newcommand{\rmh}{\mathrm{h}}
\newcommand{\rmi}{\mathrm{i}}

\begin{document}

\title{An uncertainty principle for entanglement}

\author{Mitali Nanda}
\author{Daochen Wang}
\email[Corresponding author: ]{wdaochen@gmail.com}
\affiliation{University of British Columbia}

\begin{abstract}
Consider two orthonormal bases of a bipartite Hilbert space such that states in one basis are weakly entangled and states in the other are strongly entangled. We establish a quantitative version of the following uncertainty principle: every state that is localized in one basis must be delocalized in the other. More specifically, we lower bound the sum of the Shannon entropies of any state’s coefficient distributions in the two bases by the difference between the entanglement Rényi entropies of their basis states. The result follows from new lower and upper bounds on the entanglement of superpositions that may be of independent interest. As an application, our result limits the power of weakly-entangled states in representing ground states of Hamiltonians with strongly-entangled energy eigenstates. We illustrate this application numerically for the Sachdev-Ye-Kitaev model.
\end{abstract}
\maketitle

The well-known uncertainty principle asserts that the outcomes of measuring the position and momentum of a state cannot both be very certain. In general, the same holds whenever measurement is performed in two bases that are related by a Fourier transform. In this work, we ask what happens when the two bases are arbitrary but have differing amounts of entanglement. We quantify our finding that the discrepancy in entanglement also induces uncertainty in the measurement outcomes.

For example, take any computational basis vector in $\mathbb{C}^2\otimes \mathbb{C}^2$, which is unentangled. When we measure this state in the maximally-entangled Bell basis, the outcome is uncertain. Conversely, the outcome of measuring any Bell-basis state in the standard basis is uncertain. Our main result provides a generalization of this observation to arbitrary bases. As an application, this result limits the power of weakly-entangled states in representing the ground states of Hamiltonians with strongly-entangled eigenstates, as we illustrate numerically for the Sachdev-Ye-Kitaev (SYK) model.

For a bipartite quantum state $\ket{\psi}$ and $\gamma \in [0,\infty]$, we write $\rmE_{\gamma}(\psi)$ for the entanglement $\gamma$-Rényi entropy of $\ket{\psi}$. When $\gamma = 1$, this is the usual entanglement entropy. For a probability distribution $p$, we write $\rmH(p)$ for the Shannon entropy of $p$. All logarithms have base $2$. We reserve bra-ket notation for unit vectors.

Our main result is the following theorem.

\begin{theorem}\label{thm:uncertainty}
Let $\alpha\in(1,\infty]$ and $\beta \in [1/2,1)$. Let $\{\ket{e_i}\}_i$ and $\{\ket{f_j}\}_j$ be orthonormal bases of a bipartite Hilbert space $\calH$. Suppose $\ket{\psi} \in \calH$ is a quantum state whose decompositions in these bases are given by
\begin{equation}
    \ket{\psi} = \sum_i a_i \ket{e_i} = \sum_j b_j \ket{f_j}.
\end{equation}
Then, writing $p_i \coloneqq \abs{a_i}^2$ and $q_j \coloneqq \abs{b_j}^2$, we have
\begin{equation*}
    \frac{\alpha}{\alpha-1}\rmH(p) + \frac{\beta}{1-\beta}\rmH(q) \geq  \sum_{i} p_i\, \rmE_{\alpha}(e_i) - \sum_{j} q_j\,  \rmE_{\beta}(f_j).
\end{equation*}
In particular, by setting $(\alpha,\beta) = (\infty,1/2)$, we have
\begin{equation}
    \rmH(p) + \rmH(q) \geq  \sum_i p_i\, \rmE_{\infty}(e_i) - \sum_j q_j\,  \rmE_{1/2}(f_j).
\end{equation}
\end{theorem}

We prove \cref{thm:uncertainty} by considering the entanglement of superpositions, the study of which began in Refs.~\cite{entanglement_superposition_linden_2006,entanglement_superposition_gour_2007,entanglement_superposition_gour_roy_2008}. The question is simple: how can we bound the entanglement of a superposition of orthogonal states in terms of the entanglements of the states being superposed? A natural first guess might be that the former is simply an average of the latter. However, this is not true since we can superpose unentangled states to create entanglement, and entangled states to destroy entanglement. As an example, consider superposing $\ket{00}$ and $\ket{11}$ to obtain an EPR pair. Conversely, consider superposing two Bell-basis states to obtain $\ket{00}$. Therefore, any bound must account for the extent of the superposition. This extent controls the outcome uncertainty when the superposition is measured, which is why bounds on the entanglement of superpositions can yield a result like \cref{thm:uncertainty}.

Before going further, we note that in the special case $1/\alpha + 1/\beta = 2$, \cref{thm:uncertainty} follows from the main result in Ref.~\cite{entropic_rumin_2012} and Hölder's inequality for Schatten norms~\cite[Sec.~1.1]{tqi_watrous_2018}. Compared to that result, \cref{thm:uncertainty} is more general and offers a physical interpretation in terms of entanglement. Moreover, we prove the theorem via new lower and upper bounds on the entanglement of superpositions, which may be of independent interest. Indeed, Ref.~\cite{entropic_rumin_2012} sits squarely in the literature on entropic uncertainty relations \cite{coles_entropic_2017,uncertainty_frank_2012,uncertainty_maassen_uffink_1988,uncertainty_kraus_87,uncertainty_deutsch_83}, which has thus far not been related to work on the entanglement of superpositions. 

To prove both the lower and upper bounds, we use the next lemma, which follows from interpolation theorems in functional analysis \cite{functional_analysis_stein_2011,noncommutativelp_pisier_xu_2003}. In the lemma, we write $\norm{\cdot}_r$ for the Schatten-$r$ norm and $\langle \cdot, \cdot \rangle$ for the Hilbert-Schmidt inner product, that is, $\langle A, B\rangle \coloneqq \tr[A^\dagger B]$.

\begin{restatable}{lemma}{interpolation}\label{lem:interpolation}
    Let $M_1,\dots, M_k$ be complex matrices of the same size such that $\langle M_i, M_j \rangle = \delta_{i,j}$. Let $p\in [1,\infty)\setminus\{2\}$. Then, for all $c_1,\dots,c_k \in \mathbb{C}$, and all $r$ between $p$ and $2$,  that is, $\min(p,2) \leq r \leq \max(p,2)$, we have
    \begin{equation}\label{eq:interpolation}
        \Bigl\| \sum\limits_{i=1}^k c_i M_i \Bigr\|_r \leq  \Bigl(\sum\limits_{i=1}^k \abs{c_i}^{u(r)} \cdot \norm{M_i}_p^{v(r)}\Bigr)^{1/u(r)},
    \end{equation}
    where $u(r) \coloneqq \frac{r(2-p)}{p+r-pr}$ and $v(r) \coloneqq 2-u(r)$.
\end{restatable}

\begin{proof}[Proof sketch] It is easy to see that the lemma holds at the endpoints $r\in\{p,2\}$: the case $r=2$ follows from Hilbert-Schmidt orthonormality, while the case $r=p$ follows from the triangle inequality for Schatten-$p$ norms. Appropriate interpolation gives the lemma for $r$ strictly between $p$ and $2$. See \cref{app:proof_interpolate} for the full proof.
\end{proof}

We start with the lower bound. 
\begin{proposition}\label{prop:es_lower}
    $\rmE(\psi) \geq \sum_i p_i\,  \rmE_{\alpha}(e_i) - \frac{\alpha}{\alpha-1}\rmH(p).$
\end{proposition}

\begin{proof}
    It suffices to prove the claim for finite $\alpha>1$; the case $\alpha=\infty$ follows by taking $\alpha\to\infty$.

    Let $M$ and $M_i$ be the coefficient matrices of $\vert{}\psi\rangle$ and $\vert{}e_i\rangle$, respectively, in an arbitrary but fixed product basis. Then
    \begin{equation}
    M = \sum_i a_i M_i \quad \text{and} \quad \langle M_i, M_j \rangle = \langle e_i \vert{} e_j \rangle = \delta_{ij}
    \end{equation}
    Let $p_i = \vert{}a_i\vert{}^2$. Let $\{s_\ell\}_\ell$ be the singular values of $M$. Applying \cref{lem:interpolation} with interpolation parameter $p = 2\alpha$ shows that, for all $r \in [2, 2\alpha]$,
    \begin{equation}\label{eq:lower_derivative_start}
        \sum_\ell s_\ell^r = \Vert{}M\Vert{}_r^r \le \Bigl( \sum_i p_i^{u(r)/2} \Vert{}M_i\Vert{}_{2\alpha}^{v(r)} \Bigr)^{r/u(r)}.
    \end{equation}
    Both sides equal $1$ at $r=2$. Since \cref{eq:lower_derivative_start} holds for $r\geq2$, taking right derivatives at $r=2$ and dividing by $\ln 2$ gives
    \begin{align} -\frac{1}{2} \rmE(\psi) &= \sum_\ell s_\ell^2 \log s_\ell \notag \\ &\le \frac{\alpha}{\alpha-1} \sum_i p_i \Bigl( -\frac{1}{2} \log p_i + \log \Vert{}M_i\Vert{}_{2\alpha} \Bigr) \notag \\ &= \frac{\alpha}{\alpha-1} \Bigl( \frac{1}{2} \rmH(p) + \sum_i p_i \log \Vert{}M_i\Vert{}_{2\alpha} \Bigr). \end{align}
    Since $-\log \Vert{}M_i\Vert{}_{2\alpha} = \frac{\alpha-1}{2\alpha} \rmE_\alpha(e_i)$, rearranging the preceding inequality gives the claimed result.
\end{proof}

We now turn to the upper bound.

\begin{proposition}\label{prop:es_upper}
    $\rmE(\psi) \leq \sum_j q_j\,  \rmE_{\beta}(f_j) + \frac{\beta}{1-\beta}\rmH(q).$
\end{proposition}

\begin{proof}
    Let $M$ and $M_j$ be the coefficient matrices of $\vert{}\psi\rangle$ and $\vert{}f_j\rangle$, respectively, in an arbitrary but fixed product basis. Then
    \begin{equation}
    M = \sum_j b_j M_j \quad \text{and} \quad \langle M_i, M_j \rangle = \langle f_i \vert{} f_j \rangle = \delta_{ij}.
    \end{equation}
    Let $q_j = \vert{}b_j\vert{}^2$. Let $\{s_\ell\}_\ell$ be the singular values of $M$. Applying \cref{lem:interpolation} with interpolation parameter $p = 2\beta$ shows that, for all $r \in [2\beta, 2]$,
    \begin{equation}\label{eq:upper_derivative_start}
        \sum_\ell s_\ell^r = \Vert{}M\Vert{}_r^r \le \Bigl( \sum_j q_j^{u(r)/2} \Vert{}M_j\Vert{}_{2\beta}^{v(r)} \Bigr)^{r/u(r)}.
    \end{equation}
    Both sides equal $1$ at $r=2$. Since \cref{eq:upper_derivative_start} holds for $r\leq2$, taking left derivatives at $r=2$ and dividing by $\ln 2$ gives
    \begin{align} -\frac{1}{2} \rmE(\psi) &= \sum_\ell s_\ell^2 \log s_\ell \notag \\ &\ge \frac{\beta}{1-\beta} \sum_j q_j \Bigl( \frac{1}{2} \log q_j - \log \Vert{}M_j\Vert{}_{2\beta} \Bigr) \notag \\ &= \frac{\beta}{1-\beta} \Bigl( -\frac{1}{2} \rmH(q) - \sum_j q_j \log \Vert{}M_j\Vert{}_{2\beta} \Bigr).
    \end{align}
    Since $\log \Vert{}M_j\Vert{}_{2\beta} = \frac{1-\beta}{2\beta} \rmE_\beta(f_j)$, rearranging the preceding inequality gives the claimed result.
\end{proof}

\begin{proof}[Proof of \cref{thm:uncertainty}]
Combining the inequalities in \cref{prop:es_lower,prop:es_upper} yields the theorem.
\end{proof}

\paragraph{Tightness of inequalities.} We can identify several regimes in which \cref{prop:es_lower,prop:es_upper} are tight. 

First, when $(\alpha,\beta) = (\infty, 1/2)$, take $\calH = \mathbb{C}^d \otimes \mathbb{C}^d$ and $\{\ket{e_i}\}_i$ a maximally-entangled basis and $\{\ket{f_j}\}_j$ a product-state basis, i.e., for $i=(i_1,i_2)$ and $j=(j_1,j_2)$:
\begin{equation*}
    \ket{e_i} \coloneqq d^{-1/2} \textstyle\sum_{x = 0}^{d-1} \omega^{i_1 x}\ket{x}\ket{x + i_2} \ \text{and} \ 
    \ket{f_j} \coloneqq \ket{j_1} \ket{j_2},
\end{equation*}
where $\omega \coloneqq \exp(2\pi \rmi /d)$ and $x+i_2$ is computed mod $d$, then take 
\begin{equation*}
\ket{\psi_1} \coloneqq \ket{0}\ket{0} \quad \text{and} \quad \ket{\psi_2} \coloneqq d^{-1/2} \textstyle\sum_{j=0}^{d-1} \ket{j}\ket{j}.
\end{equation*}
Then we have
\begin{align}
    \rmE(\psi_1) =&  \textstyle\sum_i p_i \, \rmE_{\infty}(e_i) - \rmH(p) = 0,
    \\
    \rmE(\psi_2) =&  \textstyle\sum_j q_j \, \rmE_{1/2}(f_j) + \rmH(q) = \log d.
\end{align}

Second, when $\alpha,\beta \to 1$, take the bases $\{\ket{e_i}\}_i$ and $\{\ket{f_j}\}_j$ to contain $\ket{\psi}$. Tightness follows since $\rmE_{\gamma} \to \rmE$  as $\gamma \to 1$. This regime raises a natural question: what happens when $\alpha=\beta = 1$? The next proposition shows that the blowup of the bounds in \cref{prop:es_lower,prop:es_upper} is unavoidable at this limit.
\begin{proposition}\label{prop:tightness}
    Fix constants $A,B>0$. There exist $\ket{\psi_1},\ket{\psi_2}\in \calH \coloneqq \mathbb{C}^d \otimes \mathbb{C}^d$ and orthonormal bases $\{\ket{e_i}\}_i$ and $\{\ket{f_j}\}_j$ of $\calH$ such that
    \begin{alignat*}{3}
        &\rmE(\psi_1) = 0 \ &&\textup{but} \ &&\sum_i p_i\rmE(e_i) - A \cdot \rmH(p) \geq \Omega(\log d),
        \\
        &\rmE(\psi_2) = \log d \ &&\textup{but} \ &&\sum_j q_j\rmE(f_j) + B \cdot \rmH(q) \leq O(\log\log d),
    \end{alignat*}
    where $p_i \coloneqq \abs{\braket{\psi_1|e_i}}^2$ and $q_j\coloneqq \abs{\braket{\psi_2|f_j}}^2$.
\end{proposition}

Note that $\log d$ is the maximum possible value for the entanglement entropy of a state in $\calH$. Therefore the proposition suggests that there is no simple modification to \cref{prop:es_lower,prop:es_upper} to allow for $\alpha=\beta = 1$.

We will prove \cref{prop:tightness} by constructing a ``small'' superposition of highly entangled states (with respect to $\rmE_1$) that is unentangled, and  a ``small'' superposition of weakly entangled states that is maximally entangled. Both constructions employ the next lemma.
\begin{lemma}\label{lem:tightness}
    Let $\ket{\psi} \in \mathbb{C}^D$. Let $\ket{u_0},\dots,\ket{u_{k-1}}\in \mathbb{C}^D$ be orthonormal. Write $\bar{u}\coloneqq k^{-1}\sum_i \ket{u_i}$.  Suppose that $\braket{u_i|\psi}$ is  independent of $i$. Then, the vectors
    \begin{equation}\label{eq:lem_tightness}
        \ket{v_i} \coloneqq  \ket{u_i} - \bar{u} + k^{-1/2}\ket{\psi}, \quad i = 0,\dots, k-1, 
    \end{equation}
    are orthonormal, and $\ket{\psi} = k^{-1/2}\sum_{i=0}^{k-1}\ket{v_i}$.
\end{lemma}
\begin{proof}
    Follows by direct calculation.
\end{proof}

\begin{proof}[Proof of \cref{prop:tightness}]
    Write $\ket{i}$ for the $(i+1)$th standard basis vector in $\mathbb{C}^d$. Write $\omega \coloneqq \exp(2\pi \rmi/d)$.  
    
    Let $k \in \{1,\dots,d\}$ be an integer to be chosen later. Then, we instantiate \cref{lem:tightness} in the following two ways.
    \begin{enumerate}[leftmargin=*]
        \item Set $\ket{\psi_1}\coloneqq \ket{0}\ket{0}$, and $\ket{u_i} \coloneqq d^{-1/2}{\sum_{j=0}^{d-1} \omega^{ij}}\ket{j}\ket{j}$ for $0\leq i \leq k-1$. 
        Clearly $\rmE(\psi_1) = 0$.
        
        Now set $\ket{e_i} \coloneqq \ket{u_i} - \bar{u} +k^{-1/2}\ket{\psi_1}$ as in \cref{eq:lem_tightness}. By \cref{lem:tightness}, these $\ket{e_i}$s are orthonormal and $\rmH(p) = \log k$. Each $\ket{e_i}$ also has high entanglement because $\abs{\braket{00|e_i}}^2 = 1/k$ and $\abs{\braket{jj|e_i}}^2 \leq 4/d$ for $j>0$, and so
        \begin{equation*}
            \rmE(e_i) \geq \sum_{j>0} \abs{\braket{jj|e_i}}^2 \log(d/4) =  (1-1/k)\log(d/4).
        \end{equation*}
        
        \item Set $\ket{\psi_2}\coloneqq d^{-1/2}\sum_{j=0}^{d-1} \ket{j}\ket{j}$, and $\ket{u_j} \coloneqq \ket{j}\ket{j}$ for $0\leq j \leq k-1$. Clearly $\rmE(\psi_2) = \log d$. 
        
         Now set $\ket{f_j} \coloneqq \ket{u_j} - \bar{u} +k^{-1/2}\ket{\psi_2}$ as in \cref{eq:lem_tightness}. By \cref{lem:tightness}, these $\ket{f_j}$s are orthonormal and $\rmH(q) = \log k$. To see that each $\ket{f_j}$ has low entanglement, we use $\rmH(r) \leq \rmh_2(1-r_{i}) + (1-r_i) \log d$ for any distribution $r$ supported on $d$ points and any index $i$, where $\rmh_2(\cdot)$ is the binary entropy function. Therefore, writing  $\delta \coloneqq 1-\abs{\braket{jj|f_j}}^2 =  1-(1-1/k+1/\sqrt{kd})^2 \leq 2/k$, we have
        \begin{equation*}
            \rmE(f_j) \leq \max_{\delta \in [0,2/k]} \bigl(\rmh_2(\delta)+ \delta \log d\bigr) \leq 1+\frac{2}{k}\log d.
        \end{equation*}
    \end{enumerate}
    Choosing $k \coloneqq \lceil \log d \rceil$ and extending $\{\ket{e_i}\}_i$ and $\{\ket{f_j}\}_j$ to orthonormal bases yields the proposition.
\end{proof}

In fact, the constructions in the proof also show that \cref{prop:es_lower,prop:es_upper} are sharp in their coefficients in front of the Shannon entropies. This is because, to leading order, we have $\rmE_\alpha(e_i) = \frac{\alpha}{\alpha-1} \log k$ and $\rmE_\beta(f_j) = \log d - \frac{\beta}{1-\beta} \log k$ by direct calculation.

\paragraph{Application to ground state computations.} 
\cref{thm:uncertainty} quantitatively limits the power of weakly-entangled states in  representing the ground state of a quantum system $H$. More specifically, by instantiating one of the bases in \cref{thm:uncertainty} as the energy-eigenstate basis of $H$ and the other as a lowly-entangled basis, the theorem implies that certain classically easy states cannot reach very low energies. Such states include matrix product states (MPSs) and configuration interaction states (CISs). For our purposes, a CIS simply refers to a linear combination of a small number of orthogonal product states.

As a case study, we apply our results to the SYK model \cite{model_sy_1993,simple_kitaev_2015,hidden_kitaev_2015}. This well-studied model captures strongly interacting fermions, and we are interested in it as computing its low-energy states may be quantumly easy~\cite{syk_anschuetz_2025}.

The SYK model for $n$ Majorana modes with $r$-body interactions ($r$ even) is defined by the Hamiltonian
\begin{equation}
H_{r}^{\mathrm{SYK}} \coloneqq i^{r/2} \binom{n}{r}^{-1/2} \sum\limits_{j_1 < \dots < j_r} g_{j_1 \dots j_r} \gamma_{j_1} \dots \gamma_{j_r}, 
\end{equation}
where $\gamma_i=\gamma_i^\dag$, $\{\gamma_i,\gamma_j\} = 2\delta_{ij}$, and the $g_{j_1\ldots j_r}$ are i.i.d. real standard Gaussians. We will illustrate our results numerically with $r=4$ and $n=24$. In this case the Hamiltonian acts on $\mathbb{C}^N$, where $N \coloneqq 2^{n/2} = 4096$.

Let $\ket{\psi}$ be either an MPS of bond dimension $D$ or a CIS with $D$ terms. In either case, the entanglement entropy $\rmE(\psi)$ of $\ket{\psi}$ across a fixed cut is upper bounded by $\log D$. For a CIS, this follows from \cref{prop:es_upper} by setting $\beta = 1/2$ and the $\ket{f_j}$s to be a product-state basis, and using $\rmH(q) \leq \log D$ \footnote{If $\ket{\psi}$ is an MPS of bond dimension $D$, the bound $\rmE(\psi) \leq \log D$ also follows from \cref{prop:es_upper} by setting $\beta$ to be any number in $[1/2,1)$ and the $\ket{f_j}$s to be a basis that contains $\ket{\psi}$, and using $\rmE_\beta(\psi) \leq \log D$. Of course, it is much easier to see $\rmE(\psi) \leq \log D$ directly in this case.}.

Now let $\ket{e_i}$s be defined as the energy eigenbasis of the Hamiltonian. Since $\rmE(\psi) \leq \log D$, \cref{prop:es_lower} gives
\begin{equation}\label{eq:entropy_lower}
    \rmH(p) \geq (1-1/\alpha) \bigl(\textstyle\sum_i p_i \rmE_{\alpha}(e_i) - \log D\bigr),
\end{equation}
which shows that $\ket{\psi}$ must be uncertain in the energy eigenbasis. In particular, this implies that the support of $p$, written $\supp(p)$, must satisfy~\footnote{This can be seen as follows. Write $s\coloneqq 1-1/\alpha$, $S\coloneqq \supp(p)$, $w_i \coloneqq 2^{-s\rmE_{\alpha}(e_i)}$, and $W\coloneqq \sum_{i\in S} w_i$. Then $\log(1/D^s) \leq \rmH(p) - \sum_i p_i s\rmE_\alpha(e_i) = \log W - \mathrm{D}(p \| (w_i/W)_i) \leq \log W$, where the first inequality is \cref{eq:entropy_lower} and the last inequality is the nonnegativity of the KL-divergence $\mathrm{D}$. Therefore $W \geq 1/D^s$.}:
\begin{equation}
\textstyle \sum_{i\in \supp(p)} 2^{-(1-1/\alpha)\rmE_\alpha(e_i)} \geq 1/D^{1-1/\alpha},
\end{equation}
and so, relabelling such that $\rmE_\alpha(e_1) \leq \dots \leq \rmE_\alpha(e_N)$, the size of $\supp(p)$ is lower bounded by
\begin{equation}\label{eq:support_lower}
\min \{k \colon \textstyle\sum_{i=1}^k 2^{-(1-1/\alpha)\rmE_\alpha(e_i)} \geq 1/D^{1-1/\alpha}\}.
\end{equation}
To help parse the preceding equation, consider the $\alpha = \infty$ case. Then the equation says the support size of $p$ is lower bounded by the number of $2^{-\rmE_{\infty}(e_i)}$'s we need to add to exceed $1/D$. If all the $\rmE_{\infty}(e_i)$'s scale as $\Omega(n)$, this yields a large lower bound on the support size, namely $2^{\Omega(n)}/D$ \footnote{Note that this lower bound still holds even if a small number of $\rmE_{\infty}(e_i)$'s scale as $\Theta(n^c)$ for fixed $0<c<1$, and $D$ is small.}. For the SYK model specifically, Ref.~\cite{entanglement_huang_2019} argues that the $\rmE(e_i)$'s scale as $\Omega(n)$, which provides some evidence that the same may hold for the $\rmE_{\infty}(e_i)$'s.

Writing $E_1 \leq \dots \leq E_N$ for the energy eigenvalues, we can also lower bound the expected energy $\langle E\rangle_\psi$ of $\ket{\psi}$ by
\begin{equation}\label{eq:energy_lower}
     \min_{p} \sum_i p_i E_i
\end{equation}
subject to \cref{eq:entropy_lower} and $p$ being a probability distribution. This optimization can be performed numerically.

To use \cref{eq:support_lower,eq:energy_lower}, we instantiate them with the eigenstate entanglement Rényi entropies and energies of $H_4^{\mathrm{SYK}}$. For the entropies, we use Rényi order $\alpha \in \{2,4,\infty\}$ as examples. These data are plotted in \cref{fig:entropy_energy}. In \cref{fig:syk_constraints}, we plot the lower bounds for the support size and the energy obtained from \cref{eq:support_lower,eq:energy_lower}.

\begin{figure}\centering\includegraphics[width=0.825\linewidth]{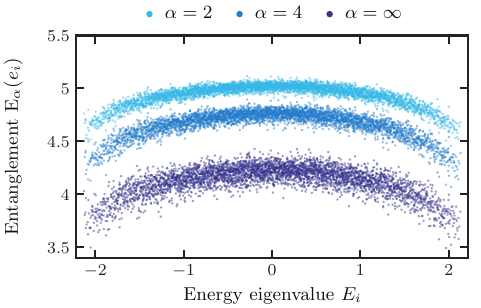}
\caption{Eigenstate entanglement Rényi  entropies versus energy for the $24$-mode $H_4^{\mathrm{SYK}}$ across a $12{:}12$ bipartition.}
\label{fig:entropy_energy}
\end{figure}

\begin{figure}\centering\includegraphics[width=0.825\linewidth]{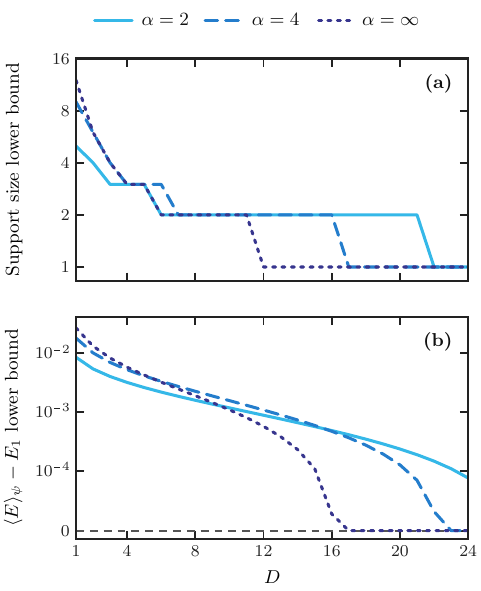}
\caption{Derived constraints on a state $\ket{\psi}$ with $\rmE(\psi) \leq \log D$ in representing the ground state of $H_4^{\mathrm{SYK}}$. (a) Lower bound on the number of energy eigenstates $\ket{\psi}$ is supported on, via \cref{eq:support_lower}. (b) Lower bound on the gap between the expected energy of $\ket{\psi}$ and the true ground state energy $E_1 = -2.127\dots$, via \cref{eq:energy_lower}.}
\label{fig:syk_constraints}
\end{figure}

We highlight that the ordering of the lines for $\alpha = 2, 4, \infty$ in \cref{fig:syk_constraints} is dependent on $D$. Since these lines represent lower bounds, the best lower bound is given by the highest line at each given $D$. The  fact that the lines are generally incomparable tells us that the bounds in \cref{prop:es_lower} for different values of $\alpha$ are incomparable.

\paragraph{Conclusion and future directions.} We showed that uncertainty in measurement outcomes in two bases can be induced by the mismatch in entanglement between the basis states. We obtained the result by proving new bounds on the entanglement of superpositions, and showed how it could be applied in a physics setting. 

Questions for future work include:
\begin{enumerate}
    \item Are there analogous uncertainty relations for other quantum resources, such as nonstabilizerness?
    \item Can we generalize our results to mixed states or multiple partitions, and is there a universal bound that subsumes ours for all values of $\alpha$ and $\beta$?
    \item Can we use these bounds to derive analytic lower bounds on, say, the entanglement entropy that a state needs to well-approximate the true ground state energy of the SYK model?
\end{enumerate}

\paragraph{Disclosure of AI use.}  
The main results of this work were conjectured and formulated by the authors, and were inspired by Refs.~\cite{uncertainty_maassen_uffink_1988,entanglement_superposition_linden_2006,entanglement_superposition_gour_2007,entanglement_superposition_gour_roy_2008,coles_entropic_2017,spectra_wang_2025}. We attribute the initial proofs of these results to ChatGPT, Claude, and Gemini. We verified these proofs and improved their presentation by simplifying and modularizing them, and by tracing their key ideas to the original sources. We used Codex to help generate and format \cref{fig:entropy_energy,fig:syk_constraints}. We used ChatGPT, Claude, and Gemini extensively to search the literature. In particular, GPT-6 Astra found that the $1/\alpha + 1/\beta = 2$ case of \cref{thm:uncertainty} follows from the main theorem in Ref.~\cite{entropic_rumin_2012} and Hölder's inequality for Schatten norms.

\bibliography{references}

\appendix
\crefalias{section}{appendix}
\onecolumngrid

\section{\texorpdfstring{Proof of \cref{lem:interpolation}}{}}\label{app:proof_interpolate}
\interpolation*

The proof of the lemma uses the three-lines lemma. For an accessible proof, see, e.g., Ref.~\cite[Lemma 2.2]{functional_analysis_stein_2011}.
\begin{lemma}[Three-lines lemma]\label{lem:three_lines}
Let $S\coloneqq \{x+iy \colon 0\leq x \leq 1, \, y\in \mathbb{R}\}$. Let $f\colon S \to \mathbb{C}$ be bounded and continuous on $S$, and analytic in the interior of $S$. Suppose $\abs{f(iy)}\leq m_0$ and $\abs{f(1+iy)}\leq m_1$ for all $y\in \mathbb{R}$. Then $\abs{f(x)} \leq m_0^{1-x} \cdot m_1^x$ for all $x\in [0,1]$.
\end{lemma}

\cref{lem:interpolation} can be viewed as a generalization of Stein's interpolation theorem given in Ref.~\cite[Theorem 2]{interpolation_stein_1956} to a linear transformation having noncommutative codomain, namely $c\mapsto \sum_i c_i M_i$. Generalization is possible because \cref{eq:variational_schatten} holds for both $\ell_s$ and Schatten-$s$ norms. Such generalization is well-known within noncommutative functional analysis, cf.~Ref.~\cite[Section 2]{noncommutativelp_pisier_xu_2003}. Related arguments have also appeared in quantum information theory, cf.~Ref.~\cite[Section 3]{trace_sutter_2016}.

\begin{proof}[Proof of \cref{lem:interpolation}]
    By absorbing any complex phase of $c_i$ into $M_i$, we may wlog assume that the $c_i$s are all nonnegative. We may wlog assume that $r$ lies strictly between $p$ and $2$, as it is easy to see that \cref{eq:interpolation} holds when $r\in \{p,2\}$. For convenience, write $m_i \coloneqq \norm{M_i}_p$. Now, let $\gamma \coloneqq \bigl(\sum_i c_i^{u(r)} m_i^{v(r)}\bigr)^{1/u(r)}$. If $\gamma = 0$, \cref{eq:interpolation} holds as both sides equal $0$. If $\gamma > 0$, by considering dividing every $c_i$ by $\gamma$, we see that we may wlog assume $\gamma = 1$. 
    
    Write $M\coloneqq \sum_i c_i M_i$, then \cref{eq:interpolation} is equivalent to $\norm{M}_r \leq 1$. To prove this, the main fact we will use is that 
    \begin{equation}\label{eq:variational_schatten}
        \norm{M}_r = \max \{ |\langle A, M\rangle|\colon  \norm{A}_{s} \leq 1 \},
    \end{equation}
    where $s$ is the Hölder conjugate of $r$, that is, $s$ is defined by $1/r + 1/s = 1$. (This appears as Eq.~(1.173) in Ref.~\cite{tqi_watrous_2018}.)

    Now fix a complex matrix $A$ with $\norm{A}_s\leq 1$. Write the SVD of $A$ as $A = \sum_i \sigma_i e_i f_i^\dagger$. For $z\in \mathbb{C}$, and with $p'$ denoting the Hölder conjugate of $p$, define
    \begin{equation}
        \quad M_z \coloneqq \sum_i c_i^{u(r) \cdot (1+z)/2} m_i^{v(r)/2 + z\cdot(v(r)/2-1)} M_i \quad \text{and} \quad A_z \coloneqq \sum_i \sigma_i^{s \cdot (\frac{1-z^*}{2} + \frac{z^*}{p'})}e_i f_i^\dagger,
    \end{equation}
    where $z^*$ denotes the complex conjugate of $z$,
    and consider the map 
    \begin{equation}
        f \colon \mathbb{C} \to \mathbb{C}; \quad z \mapsto \langle A_{z}, \,  M_z\rangle.
    \end{equation}
    
    Write $x \coloneqq \frac{p(2-r)}{r(2-p)}\in (0,1)$ so that $\frac{1}{r} = \frac{1-x}{2} + \frac{x}{p}$, $u(r) = \frac{2}{1+x}$, and $v(r) = x u(r)$. Then the definitions imply $f(x) = \langle A, M\rangle$. To bound $\abs{f(x)}$, we will bound $f$ on the two lines $L_0 \coloneqq \{iy \colon y \in \mathbb{R}\}$ and  $L_1 \coloneqq \{1+iy \colon y \in \mathbb{R}\}$, and then use \cref{lem:three_lines}. (The lemma applies since it is easy to check that $f$ is bounded and continuous on $S\coloneqq \{x+iy \colon 0\leq x \leq 1, \, y\in \mathbb{R}\}$, and analytic in the interior of $S$.)

    \begin{enumerate}
        \item For $z \in L_0$, we can use $\abs{f(z)} = | \langle A_z, \,  M_z\rangle| \leq \norm{A_z}_2\, \norm{M_z}_2$ and $\sum_i \sigma_i^s \leq 1$ and $\gamma = 1$ to see $\abs{f(z)}\leq 1$.
        \item For $z\in L_1$, we can use $\abs{f(z)} = |\langle A_z, \,  M_z\rangle| \leq \norm{A_z}_{p'} \, \norm{M_z}_p$ and $\sum_i \sigma_i^s \leq 1$ and $\gamma = 1$ to see $\abs{f(z)}\leq 1$.
    \end{enumerate}
    
    Therefore, \cref{lem:three_lines} implies $\abs{f(x)} \leq 1$. Since $A$ is arbitrary, \cref{eq:variational_schatten} implies that $\norm{M}_r \leq 1$, as required.
\end{proof}

\end{document}